\documentclass[runningheads]{llncs}

\usepackage[utf8]{inputenc}
\usepackage[T1]{fontenc}

\usepackage{graphicx}
\usepackage[svgnames]{xcolor}
\usepackage{fca}
\usepackage{mathtools}
\usepackage{marvosym}
\usepackage{cite}
\usepackage[hidelinks]{hyperref}

\usepackage{enumitem}
\usepackage{booktabs}
\usepackage{tikz}
\usepackage{tikzscale}
\usepackage[outline]{contour}   
\contourlength{2.5pt}  
\dimendef\prevdepth=0    
\usepackage{svg}
\usepackage[ruled,vlined]{algorithm2e}
\usepackage{cleveref}

\newcommand{\todo}[1]{{\bf\color{red} Todo: #1}}

\newcommand{\old}[1]{{\color{teal} #1}}

\renewcommand{\old}[1]{ }
\newcommand{\commentout}[1]{ }

\newcommand{\ie}{i.\,e.\xspace}

\newcommand{\K}{\mathbb{K}}

\newcommand{\BB}{\mathfrak{B}}

\newcommand{\UL}{\underline{L}}
\newcommand{\US}{\underline{S}}

\newcommand{\UBB}{\underline{\BB}}
\DeclareMathOperator{\supprime}{supremum-prime}
\DeclareMathOperator{\infprime}{infimum-prime}
\newcommand{\contradiction}{(\text{\Lightning})}
\DeclareRobustCommand{\neswarrow}{%
	\mathrel{\text{\ooalign{$\swarrow$\cr$\nearrow$}}}%
}
\renewcommand{\intent}[1]{\text{intent}(#1)}
\let\intent\relax
\DeclareMathOperator{\intent}{intent}

\let\subset\subseteq
\let\cref\Cref

\begin{document}
\title{Dismantlable Intervals}
\title{Interval-Dismantling in Concept Lattices}
\title{Dismantling by Intervals}
\title{Interval-Dismantling for Lattices}
%
%
\author{Maximilian Felde\inst{1,2}\orcidID{0000-0002-6253-9007} \and\\
{Maren Koyda\inst{1,2}\orcidID{0000-0002-8903-6960}}}
%
%
\institute{%
  Knowledge \& Data Engineering Group,
  University of Kassel, Germany
  \and
  Interdisciplinary Research Center for Information System Design
  \mbox{University of Kassel, Germany}
  \email{felde@cs.uni-kassel.de, koyda@cs.uni-kassel.de}
}
%
\maketitle              
\begin{abstract}

  Dismantling allows for the removal of elements of a set, or in our case lattice, without disturbing the remaining structure.
  In this paper we have extended the notion of dismantling by single elements to the dismantling by intervals in a lattice.
  We utilize theory from Formal Concept Analysis (FCA) to show that lattices dismantled by intervals correspond to closed subrelations in the respective formal context, and that there exists a unique kernel with respect to dismantling by intervals.
  Furthermore, we show that dismantling intervals can be identified directly in the formal context utilizing a characterization via arrow relations and provide an algorithm to compute all dismantling intervals.


\end{abstract}

\section{Introduction}
\label{sec:introduction}

The dismantling of elements in ordered sets, in particular that of irreducible elements, was examined for example in \cite{duffus1976crowns,rival1974lattices,kelly1974crowns,baker1972partial,WALKER1984275}.
Notably, in \cite{duffus1976crowns} Duffus and Rival examined the dismantling of doubly irreducible elements and proved the uniqueness of the DI-core, \ie, the poset that remains when all doubly irreducible elements have been removed. 
In \cite{farley1993uniqueness} Farley gave a simpler proof relying on results for semimodular posets.

In Formal Concept Analysis (FCA) the removal of a doubly irreducible element in a concept lattice corresponds to the removal of a single incidence (cross) from its (clarified) formal context, cf. \cite[Prop. 53]{fca-book}.
The dismantling of a doubly irreducible element results in a complete sublattice of the original concept lattice.
In particular, this sublattice contains all but one of the original concepts.

In this paper, we extend the notion of dismantling from single elements to intervals in order to remove multiple (not necessarily irreducible) concepts at once while preserving the remaining concept lattice.
To this end, we make use of the one-to-one correspondence between closed subrelations of a formal context and the complete sublattices of its concept lattice~\cite{Wille87b}, and, more generally, of the one-to-one correspondence between closed subcontexts of a context and the sublattices of its concept lattice~\cite{koyda2021boolean}.

Extending dismantling to intervals, we call an interval $[u,v]\eqqcolon \US$ \emph{dismantling for a lattice $\UL$} if $v$ is infimum prime in the filter of $u$, $u$ is supremum prime in the ideal of $v$ and $u,v\not\in\{\bot,\top\}$.
Because infimum (supremum) prime implies infimum (supremum) irreducible, dismantling intervals that consist of a single element are precisely the doubly irreducible elements.
We show that an interval $\US$ is dismantling for $\UL$ if and only if the incidences of all concepts not in $\US$ form a closed subrelation.
Furthermore, we show that the core obtained by the iterative removal of dismantling intervals for a lattice is unique.
Where possible, we use the more general notion of an interval $\US$ being \emph{quasi-dismantling for a lattice}, which allows for $u,v\in\{\bot,\top\}$ and show that an interval $\US$ is quasi-dismantling for $\UL$ if and only if the objects, attributes and incidences of all concepts not in $\US$ form a closed subcontext.

Finally, we give a characterization of dismantling intervals via arrow relations on the context side and provide an algorithm to determine whether an interval is dismantling using this characterization.
Furthermore, the arrow relations provide a way to compute all dismantling intervals for a given context $\K$ without having to compute the concept lattice $\UBB(\K)$ itself.

This paper is structured as follows:
In \Cref{sec:fca-basics} we recollect all required notions from order theory and FCA.
In \Cref{sec:dism-interv} we introduce \emph{dismantling intervals}, develop their theory and prove the one-to-one correspondence between dismantling intervals and closed subrelations.
In \Cref{sec:dismantling-contexts-and-arrow-relations} we characterize dismantling intervals via arrow relations on the context side and give an algorithm to compute all dismantling intervals.
In \Cref{sec:conclusion} we give a brief conclusion.

\section{Basics}
\label{sec:fca-basics}

In the following we recall some notions from order theory, cf.~\cite{Birkhoff}, and formal concept analysis, cf.~\cite{fca-book}, and introduce some notations used in this work.

An \emph{ordered set} (or short \emph{order}) is a tuple $\underline{L}\coloneqq (L,\le)$ consisting of a set $L$ and a reflexive, transitive and antisymmetric relation $\le\subseteq L\times L$.
A subset $S$ of an order $\UL$ together with the same order relation restricted to $S$, \ie, $\US\coloneqq(\US, \leq\!\! \cap S^2)$, is called \emph{suborder} of $\UL$.
We denote this by $\US\le \UL$.
Specific suborders are generated by a single element $c\in L$:
The \emph{ideal} of $c$ is defined as $(c]:=\{x\in L\mid x\le c\}$;
The \emph{filter} of $c$ is defined as $[c):=\{x\in L\mid c\le x\}$.
For two elements $c,d \in L$ with $c\le d$ the \emph{interval} between $c$ and $d$ is given by $[c,d]:=\{x\in L\mid c\le x\le d\}$.
Further, we define the order $\UL\setminus\US\coloneqq (L\setminus S, (L\setminus S)^2 \cap \leq)$.
An element $c\in \UL$ of an ordered set is called \emph{upper bound} for a subset $T\subseteq L$ if $c\ge x$ for all $x\in T$.
If $c$ is the unique smallest upper bound of $T$, we call $c$ the \emph{supremum} of $T$.
Analogous, an element $c\in \UL$ of an ordered set is called \emph{lower bound} for a subset $T\subseteq L$ if $c\le x$ for all $x\in T$.
If $c$ is the unique greatest lower bound of $T$, we call $c$ the \emph{infimum} of $T$.

An ordered set $\UL$ is called a \emph{lattice} if for any two elements $c,d\in\UL$ there is an infimum $c\wedge d$ and a supremum $c\vee d$ in $\UL$.
It is called a \emph{complete lattice} if  all subsets $X\subseteq L$ have an infimum $\bigwedge X$ and a supremum $\bigvee X$ in $L$.
If $\UL$ is a lattice, we call $\US$ a \emph{sublattice} of $\UL$ if $(a,b\in S\Rightarrow a\vee b\in S \text{ and } a\wedge b\in S)$ holds.
We call $\US$ a \emph{complete sublattice} if for every subeset $T\subseteq S$ also $\bigvee T\in S$ and $\bigwedge T\in S$.
In finite lattices this requirement translates into top $(\top)$ and bottom $(\bot)$ element of $\UL$ being included in $S$.
For an element $c\in\UL$ we define $c_*\coloneqq\bigvee\{x\in \UL\mid x< c\}$ and 
$c^*\coloneqq\bigwedge\{x\in \UL\mid x> c\}$.
We call $c\in \UL$ \emph{supremum-irreducible} if $c$ has exactly one lower neighbor, meaning $c_*<c$.
An element $c\in \UL$ is called \emph{infimum-irreducible} if $c$ has exactly one upper neighbor, meaning $c^*>c$.
We call $c$ \emph{doubly-irreducible} if it is both, supremum-irreducible and infimum-irreducible.
An element $c\in\UL$ is called \emph{supremum-prime} if for all $x,y\in\UL$: $c \leq x\vee y \Rightarrow c\leq x$ or $c \leq y$.
Analogously, an element $c\in\UL$ is called \emph{infimum-prime} if for all $x,y\in\UL$: $x\wedge y \leq c \Rightarrow x\leq c$ or $y\leq c$.
It follows, an element that is supremum-prime (infimum-prime) is also supremum-irreducible (infimum-irreducible).

In the field of FCA one fundamental structure is the \emph{formal context}.
A \emph{formal context} $\K=\GMI$ is a triple consisting of a \emph{object set} $G$ an \emph{attribute set} $M$ and an \emph{incidence relation} $I\subseteq G\times M$.
On the powerset of the objects and the powerset of the attributes the \emph{derivation operators} are defined as follows:
$\cdot': \mathcal{P}(G)\rightarrow \mathcal{P}(M), A\mapsto A'\coloneqq \{m\in M\mid \forall g\in A: (g,m)\in I\}$ and 
$\cdot': \mathcal{P}(M)\rightarrow \mathcal{P}(G), B\mapsto B'\coloneqq \{g\in G\mid \forall m\in B: (g,m)\in I\}$.
A pair $c=(A,B)$ with $A\subseteq G$ and $B\subseteq M$ is called \emph{formal concept} of the formal context $\K=\GMI$ if $A'=B$ and $B'=A$ hold.
Then, $A$ is called the \emph{extent} and $B$ is called the \emph{intent} of $c$.
For an object $g\in G$ the concept $\gamma g \coloneqq (g'',g')$ is called \emph{object concept} of $g$.
Analogous, the concept $\mu m \coloneqq (m',m'')$ is called \emph{attribute concept} of the attribute $m\in M$.
The set of all formal concepts of a formal context $\K$ is denoted by $\BB(\K)$.
Together with the order relation $\le$, where $(A_1,B_1)\le(A_2,B_2)$ if $A_1\subseteq A_2$, $\BB(\K)$ forms a lattice, the \emph{concept lattice} $\underline{\BB}(\K)$.

An object $g\in G$ is called \emph{irreducible}, if there is no set of objects $X\subset G$ with $g\not\in X$ and $X'=g'$.
Otherwise, $g$ is called \emph{reducible}.
Analogous, an attribute $m\in M$ is called \emph{irreducible}, if there is no set of attributes $X\subset M$ with $m\not\in X$ and $X'=m'$. 
Otherwise $m$ is called \emph{reducible}.
If all objects and all attributes of a formal context $\K$ are irreducible, $\K$ is called \emph{reduced}. 
The concept lattice of a reduced formal context is isomorphic to the concept lattice of the context with additional reducible objects or attributes.
In a reduced context $\K=\GMI$ the object concepts $\gamma g$ are the supremum-irreducible concepts and the attribute-concepts $\mu m$ are the infimum-irreducible concepts in $\UBB(\K)$.

For a formal context $\K=\GMI$ a \emph{subrelation} of $I$ is a subset $J\subseteq I$.
$J$ is called \emph{closed subrelation} if all formal concepts of $(G,M,J)$ are formal concepts of $\K$ as well, \ie, if $\BB(G,M,J)\subseteq \BB(\K)$. 
A formal context $\mathbb{S}=\K|_{H,N}\coloneqq(H,N,J)$ with $H\subseteq G$, $N\subseteq M$ and $J=I\cap (H\times N)$ is called \emph{subcontext} of $\K=(G,M,I)$ and denoted by $\mathbb{S}\le\K$.
If instead $J\subseteq I\cap (H\times N)$ and all formal concepts of $\mathbb{S}$ are also formal concepts of $\K$, $\mathbb{S}$ is called \emph{closed subcontext} of $\K$.
The closed subrelations of a formal context $\K$ are in a one-to-one correspondence to the complete sublattices of $\underline{\BB}(\K)$ \cite{Wille87b}:
If $J$ is a closed relation of $\GMI$, then $\UBB(G,M,J)$ is a complete sublattice of $\UBB\GMI$ with $J=\bigcup\{A\times B\ |\ (A,B)\in\UBB(G,M,J)\}$.
Conversely, for every complete sublattice $U$ of $\UBB\GMI$ the relation $J\coloneqq \bigcup\{A\times B\ |\  (A,B)\in U \}$ is closed and $\UBB(G,M,J)=U$.
Similarly, the closed subcontexts of a finite formal context $\K$ are in a one-to-one correspondence to the sublattices of $\underline{\BB}(\K)$ \cite{koyda2021boolean}:
If $(H,N,J)$ is a closed subcontext of $\GMI$, then $\UBB(H,N,J)$ is a sublattice of $\UBB\GMI$ with $H=\bigcup\{A\ |\ (A,B)\in\UBB(H,N,J)\}$, $N=\bigcup\{B\ |\ (A,B)\in\UBB(H,N,J)\}$ and $J=\bigcup\{A\times B\ |\ (A,B)\in\UBB(H,N,J)\}$.
Conversely, for every sublattice $U$ of $\UBB\GMI$ the context $(H,N,J)=(\bigcup_{(A,B)\in U} A, \bigcup_{(A,B)\in U} B, \bigcup_{(A,B)\in U} A\times B)$ is a closed subcontext of $\GMI$ and $\UBB(H,N,J)=U$.
In the following all sets are considered finite.

\section{Dismantling Intervals for a Lattice}
\label{sec:dism-interv}

To identify the intervals that can be removed from a (concept) lattice without disturbing the remaining structure, we introduce the notions of \emph{dismantling} and \emph{quasi-dismantling intervals for a lattice}, extending the usual notion of dismantling single elements.

\begin{definition}\label{def:dismanting}
	Let $\UL$ be a lattice and 
	$[u,v]=\US\leq\UL$ an interval of $\UL$.
	We call $\US$ \emph{quasi-dismantling} for $\UL$ if $u$ is $\supprime$ in $(v]$ and $v$ is $\infprime$ in $[u)$.
  If $u\not= \bot$ and $v\not= \top$ we call $\US$ \emph{dismantling} for $\UL$.
\end{definition}

See \cref{fig:dismantlable-viz} for a visualization and an example for \cref{def:dismanting}.

On the context side the removal of a set of concepts $S$ corresponds to the removal of all incidences (crosses) that only belong to concepts in the respective interval.
We call the remaining context (objects, attributes, incidences) the \emph{S-removed} context (objects, attributes, incidences).

\begin{figure}[t]
  \centering
  	\begin{minipage}{0.4\textwidth}
    \includeinkscape[width=0.55\columnwidth]{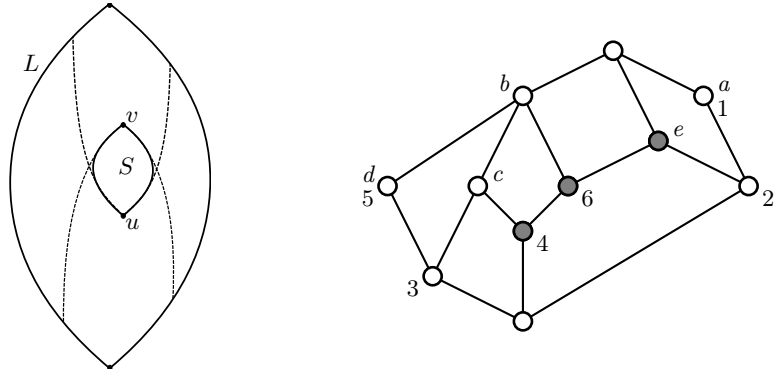}
  \end{minipage}
  \begin{minipage}{0.4\textwidth}
		\centering
      	{\unitlength 0.6mm
		\begin{picture}(80,60)%
		\put(0,0){%
			\begin{diagram}{80}{60}
			\Node{ 1}{30}{ 0}
			\Node{ 2}{10}{10}
			\Node{ 3}{30}{20}
      \Node{ 4}{20}{30}
      \Node{ 5}{ 0}{30}
      \Node{ 6}{40}{30}     
      \Node{ 7}{80}{30}     
      \Node{ 8}{60}{40}     
      \Node{ 9}{70}{50}
      \Node{10}{30}{50}     
      \Node{11}{50}{60}   
      \Edge{1}{2}
      \Edge{1}{3}
      \Edge{1}{7}
      \Edge{2}{4}
      \Edge{2}{5}
      \Edge{3}{4}
      \Edge{3}{6}
      \Edge{4}{10}
      \Edge{5}{10}     
      \Edge{6}{10}
      \Edge{6}{8}
      \Edge{7}{8}
      \Edge{7}{9}
      \Edge{8}{11}      
      \Edge{9}{11}
      \Edge{10}{11}
      \leftObjbox{2}{3}{1}{3}
      \leftObjbox{5}{3}{1}{5}
      \rightObjbox{3}{3}{1}{4}
      \rightObjbox{7}{3}{1}{2}
      \rightObjbox{6}{3}{1}{6}      
      \rightObjbox{9}{3}{1}{1}
			\leftAttbox{5}{3}{1}{d}
	 	  \leftAttbox{10}{3}{1}{b}
			\rightAttbox{8}{3}{1}{e}
	 	  \rightAttbox{4}{3}{1}{c}
	 	  \rightAttbox{9}{3}{1}{a}
    \end{diagram}}
			\put(30,20){\ColorNode{gray}}
			\put(40,30){\ColorNode{gray}}
			\put(60,40){\ColorNode{gray}}
		\end{picture}}	
	\end{minipage}
  \caption{Visualization for the definition of an interval $\US=[u,v]$ dismantling for $\UL$ (left). The colored elements are an example for an interval dismantling for the lattice (right).}
    \label{fig:dismantlable-viz}
\end{figure} 

\begin{definition}
	Let $\K=\GMI$ be a formal context and $S\subseteq \BB(\K)$ a set of formal concepts.
	We call
  \begin{itemize}
  \item 
	\(
	I_S \coloneqq I\setminus \left(\bigcup_{(A,B)\in S}A\times B \setminus \bigcup_{(A,B)\in\BB(\K)\setminus S}A\times B \right) 
	\)
	\emph{S-removed incidences},
\item 
	\(
	G_S \coloneqq G\setminus \left(\bigcup_{(A,B)\in S}A \setminus \bigcup_{(A,B)\in\BB(\K)\setminus S} A\right) 
	\)
	\emph{S-removed objects},
\item 
	\(
	M_S \coloneqq M\setminus \left(\bigcup_{(A,B)\in S}B \setminus \bigcup_{(A,B)\in\BB(\K)\setminus S}B \right) 
	\)
	\emph{S-removed attributes}.
    \end{itemize}
	In the following we let $\K_S\coloneqq (G_S,M_S,I_S)$ be the \emph{S-removed context} for $\K$.
\end{definition}

There is a simpler representation for the S-removed objects, attributes and incidences:

\begin{lemma}
		Let $\K=\GMI$ be a formal context and $S\subseteq \BB(\K)$ a set of formal concepts, then
		$I_S= \bigcup_{(A,B)\in\BB(\K)\setminus S} A\times B$,
		$G_S= \bigcup_{(A,B)\in\BB(\K)\setminus S} A$ and
		$M_S= \bigcup_{(A,B)\in\BB(\K)\setminus S} B.$
	\end{lemma}

\begin{proof}
	We show the proof for the S-removed incidences. The proofs for the other two sets are analogous.

\noindent$I\setminus \left(\bigcup_{(A,B)\in S}A\times B \setminus \bigcup_{(A,B)\in\BB(\K)\setminus S}A\times B \right)$\\
    $= \bigcup_{(A,B)\in\BB(\K)}A\times B \setminus \left(\bigcup_{(A,B)\in S}A\times B \setminus \bigcup_{(A,B)\in\BB(\K)\setminus S}A\times B \right)
    \\=  \left(\bigcup_{(A,B)\in\BB(\K)}A\times B \cap \left(\bigcup_{(A,B)\in S}A\times B\right)^{\! c}\   \right)$\\%
  \phantom{= }$\cup \left(\bigcup_{(A,B)\in\BB(\K)}A\times B \cap \bigcup_{(A,B)\in\BB(\K)\setminus S}A\times B  \right)$\\
$= \bigcup_{(A,B)\in\BB(\K)\setminus S}A\times B$     
  \qed
\end{proof}

If we consider an interval $\US$, we see that $\US$ being quasi-dismantling corresponds to obtaining a closed subcontext on S-removal.
More precisely, for an interval $\US\leq \UBB(\K)$ the S-removed context $\K_S$ for $\K$ is a closed subcontext if and only if $\US$ is quasi-dismantling for $\UBB(\K)$.

  \begin{lemma}
    \label{lem:closed_subrelation}
	Let $\K=\GMI$ be a formal context and $\UBB(\K)$ its corresponding concept lattice.
	Let $\US=[u,v]\leq\UBB(\K)$ be an interval.
	Then, $\US$ is quasi-dismantling for $\UBB(\K)$ iff $\K_S= (G_S,M_S,I_S)$ is a closed subcontext of $\K$.
\end{lemma}
\begin{proof}
	"$\Rightarrow$": We show the contraposition:
	Assume $(G_S,M_S,I_S)$ is no closed subcontext.
	By definition holds $G_S\subseteq G$, $M_S\subseteq M$ and $I_S\subseteq I$.
	Then there exists some $c\in\UBB(\K_S)$ such that $c\not\in\UBB(\K)$.
	Since $\UBB(\K_S)$ is a lattice generated from $\BB(\K)\setminus S$ there exist $x,y\in \BB(\K)\setminus S$ such that $x\vee y = c$ or $x \wedge y = c$ in $\UBB(\K_S)$.
	In case $x\vee y = c$:
	Since $\UBB(\K)$ is a lattice it follows that there exists some $z\in \BB(\K)$ with $z\not\in\BB(\K_S)$ and $z=x\vee y$.
	Thus, $z\in S = [u,v]$ and therefore $z\geq u$ in $\UBB(\K)$.
	Because $x,y\not\in S$ we have $x,y\not\geq u$.
	Hence, $u$ is not $\supprime$ in $(v]$ and $S$ is not quasi-dismantling.
	The case $x\wedge y = c$ is analogous.
  
	"$\Leftarrow$": We show the contraposition: 
	Assume $\US$ is not quasi-dismantling.
	Then $u$ is not supremum-prime in $(v]$ or $v$ is not infimum-prime in $[u)$.
	In case that $u$ is not supremum-prime in $(v]$:
	There exist $x,y\in (v]$ such that $z:=x \vee y\geq u$, $x\not\geq u$ and $y\not\geq u$.
	Thus, $x,y\not\in \US$ and $z\in \US$.
	Therefore, $z\not\in \BB(\K_S)$.
	There is some supremum $c=x\vee y$ in $\BB(\K_S)$.
	Because the intent of $c$ is the intent of $z$ by~\cite[Thm. 3]{fca-book} we have $c\not\in \BB(\K)$.
	Thus, $(G_S,M_S,I_S)$ is no closed subcontext.
	The case that $v$ is not infimum-prime in $[u)$ is analogous.
	\qed
\end{proof}

In particular, for a dismantling interval $\US$ we have $G_S=G$ and $M_S=M$ (because $\top,\bot\in\BB(\K)\setminus S$) and therefore the correspondence in this case is to closed subrelations.

The removal of a quasi-dismantling interval leaves the remaining lattice intact with respect to supremum and infimum:

\begin{lemma}	\label{lem:lattice}
	Let $\UL$ be a lattice and $\US\le\UL$ an interval.
	If $\US$ is quasi-dismantling for $\UL$, then $\UL\setminus\US$ is a lattice.
	In particular, $\UL\setminus\US$ is a sublattice of $\UL$.
\end{lemma}
\begin{proof}
	Let $x,y, z\in \UL$ with $z=x\vee y$.
	Because $\US$ is quasi-dismantling if $x,y\not\in\US$ then $z\not\in\US$.
	Analogously for $z=x\wedge y$.	
	\qed
\end{proof}


\begin{corollary}
  If $\US$ is dismantling for $\UL$ then $\UL\setminus\US$ is a complete sublattice of $\UL$.
\end{corollary}

Combining the previous lemmas, it follows for an interval $\US$ which is quasi-dismantling in a lattice $\UBB(\K)$ that the removal of $\US$ from $\UBB(\K)$ is isomorphic to the concept lattice of the S-removed context $\K_S$.

\begin{theorem}\label{thm:main}
	Let $\K=\GMI$ be a formal context and $\UBB(\K)$ its corresponding concept lattice.
	Let $\US\leq\UBB(\K)$ be an interval.
	If $\US$ is quasi-dismantling for $\UL$, then 
	\[ \UBB(\K)\setminus\US = \UBB(\K_S).
	\]
\end{theorem}
\begin{proof}
	We know from~\cref{lem:lattice} that $\UBB(\K)\setminus\US$ is a sublattice of $\UBB(\K)$.
	Further, from~\cref{lem:closed_subrelation} follows that $\UBB(\K_S)$ is a sublattice of $\UBB(\K)$.
	Both contain exactly the concepts of $\BB(\K)$ that are not included in $\US$.
	\qed
\end{proof}

An example for a lattice with a dismantling interval is given in \cref{fig:dismantlable}.
The concept lattice of $\K_S$ is the concept lattice of $\K$ without the interval $\US$.

Note that this statement does not hold for intervals that are not dismantling for the lattice.
See \cref{fig:not_dismantlable} for a counterexample.

\begin{figure}[t]
  \centering
	\begin{minipage}{0.35\textwidth}
		\centering
		{\unitlength 0.6mm
			\begin{picture}(60,45)%
			\put(0,0){%
				\begin{diagram}{60}{45}
				\Node{1}{30}{0}
				\Node{2}{15}{15}
				\Node{3}{30}{15}
				\Node{4}{45}{15}
				\Node{5}{15}{30}
				\Node{6}{30}{30}
				\Node{7}{45}{30}
				\Node{8}{30}{45}
				\Edge{1}{2}
				\Edge{1}{3}
				\Edge{1}{4}
				\Edge{2}{5}
				\Edge{2}{6}
				\Edge{3}{5}
				\Edge{3}{7}
				\Edge{4}{6}
				\Edge{4}{7}
				\Edge{5}{8}
				\Edge{6}{8}
				\Edge{7}{8}
				\leftObjbox{2}{3}{1}{1}
				\rightObjbox{3}{3}{1}{2}
				\rightObjbox{4}{3}{1}{3}
				\leftAttbox{5}{3}{1}{a}
				\rightAttbox{6}{3}{1}{b}
				\rightAttbox{7}{3}{1}{c}
				\end{diagram}}
			\put(15,15){\ColorNode{gray}}
			\put(15,30){\ColorNode{gray}}
			\end{picture}}			
	\end{minipage}
	\begin{minipage}{0.3\textwidth}
		\centering
		\begin{cxt}%
			\att{a}%
			\att{b}%
			\att{c}%
			\obj{xc.}{1} %
			\obj{c.c}{2} %
			\obj{.cc}{3} %
		\end{cxt}
	\end{minipage}
	\begin{minipage}{0.3\textwidth}
	\centering
	{\unitlength 0.6mm
		\begin{picture}(60,45)%
		\put(0,0){%
			\begin{diagram}{60}{45}
			\Node{1}{30}{0}
			\Node{2}{30}{15}
			\Node{3}{45}{15}
			\Node{4}{30}{30}
			\Node{5}{45}{30}
			\Node{6}{30}{45}
			\Edge{1}{2}
			\Edge{1}{3}
			\Edge{2}{5}
			\Edge{3}{5}
			\Edge{4}{3}
			\Edge{4}{6}
			\Edge{5}{6}
			\leftObjbox{2}{3}{1}{2}
			\rightObjbox{3}{3}{1}{3}
			\leftObjbox{4}{3}{1}{1}
			\leftAttbox{4}{3}{1}{b}
			\rightAttbox{5}{3}{1}{c}	
			\leftAttbox{2}{3}{1}{a}	
			\end{diagram}}	
		\end{picture}}			
\end{minipage}
\caption{Example of a lattice (left) and the corresponding formal context $\K={\GMI}$ (middle).
		The highlighted interval $\US$ is dismantling for the lattice.
		Therefore, the highlighted $S$-removed incidences of $\K$ are a closed subrelation of $I$.
		The lattice on the right corresponds to the context $\K_S$.
	}
	\label{fig:dismantlable}
\end{figure}
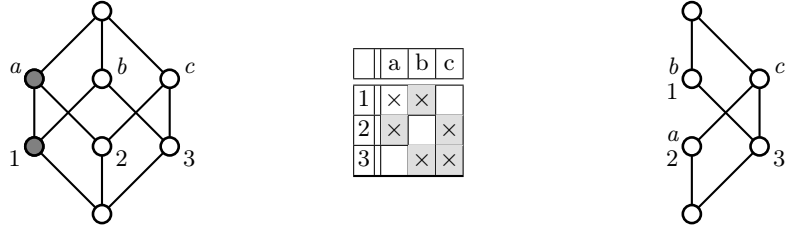

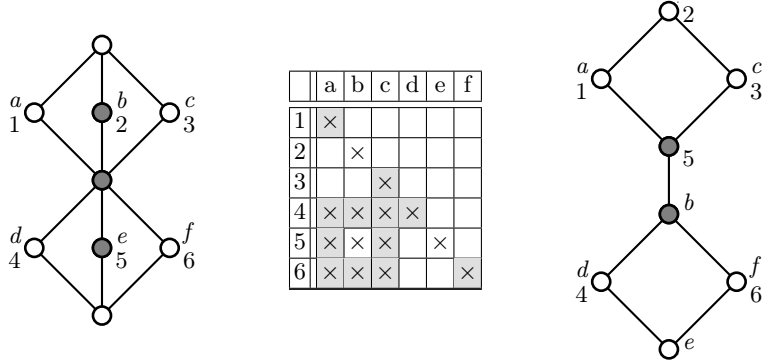
\begin{figure}[t]
  \centering
		\begin{minipage}{0.3\textwidth}
		\centering
		{\unitlength 0.6mm
			\begin{picture}(60,60)%
			\put(0,0){%
				\begin{diagram}{60}{60}
				\Node{1}{30}{0}
				\Node{2}{15}{15}
				\Node{3}{30}{15}
				\Node{4}{45}{15}
				\Node{5}{30}{30}
				\Node{6}{15}{45}
				\Node{7}{30}{45}
				\Node{8}{45}{45}
				\Node{9}{30}{60}
				\Edge{1}{2}
				\Edge{1}{3}
				\Edge{1}{4}
				\Edge{2}{5}
				\Edge{3}{5}
				\Edge{4}{5}
				\Edge{5}{6}
				\Edge{5}{7}
				\Edge{5}{8}
				\Edge{6}{9}
				\Edge{7}{9}
				\Edge{8}{9}
				\leftObjbox{2}{3}{1}{4}
				\rightObjbox{3}{3}{1}{5}
				\rightObjbox{4}{3}{1}{6}
				\leftObjbox{6}{3}{1}{1}
				\rightObjbox{7}{3}{1}{2}
				\rightObjbox{8}{3}{1}{3}
				\leftAttbox{2}{3}{1}{d}
				\rightAttbox{3}{3}{1}{e}
				\rightAttbox{4}{3}{1}{f}
				\leftAttbox{6}{3}{1}{a}
				\rightAttbox{7}{3}{1}{b}
				\rightAttbox{8}{3}{1}{c}
				\end{diagram}}
			\put(30,15){\ColorNode{gray}}
			\put(30,30){\ColorNode{gray}}
			\put(30,45){\ColorNode{gray}}
			\end{picture}}			
	\end{minipage}
	\begin{minipage}{0.3\textwidth}
		\centering
		\begin{cxt}%
			\att{a}%
			\att{b}%
			\att{c}%
			\att{d}%
			\att{e}%
			\att{f}%
			\obj{c.....}{1} %
			\obj{.x....}{2} %
			\obj{..c...}{3} %
			\obj{cccc..}{4} %
			\obj{cxc.x.}{5} %
			\obj{ccc..c}{6} %
		\end{cxt}
	\end{minipage}
\begin{minipage}{0.3\textwidth}
	\centering
	{\unitlength 0.6mm
		\begin{picture}(60,75)%
		\put(0,0){%
			\begin{diagram}{60}{75}
			\Node{1}{30}{0}
			\Node{2}{15}{15}
			\Node{3}{45}{15}
			\Node{4}{30}{30}
			\Node{5}{30}{45}
			\Node{6}{15}{60}
			\Node{7}{45}{60}
			\Node{8}{30}{75}
			\Edge{1}{2}
			\Edge{1}{3}
			\Edge{2}{4}
			\Edge{3}{4}
			\Edge{4}{5}
			\Edge{5}{6}
			\Edge{5}{7}
			\Edge{6}{8}
			\Edge{8}{7}
			\leftObjbox{2}{3}{1}{4}
			\rightObjbox{3}{3}{1}{6}
			\rightObjbox{5}{3}{1}{5}
			\leftObjbox{6}{3}{1}{1}
			\rightObjbox{7}{3}{1}{3}
			\rightObjbox{8}{3}{-1}{2}
			\leftAttbox{2}{3}{1}{d}
			\rightAttbox{3}{3}{1}{f}
			\rightAttbox{4}{3}{1}{b}
			\leftAttbox{6}{3}{1}{a}
			\rightAttbox{7}{3}{1}{c}
			\rightAttbox{1}{3}{0}{e}
			\end{diagram}}
		\put(30,45){\ColorNode{gray}}
		\put(30,30){\ColorNode{gray}}
		\end{picture}}			
\end{minipage}
\caption{Example of a lattice (left) and the corresponding formal context $\K=\GMI$ (middle).
		The highlighted interval $\US$ is not dismantling for the lattice.
		Therefore, the highlighted $S$-removed incidences of $\K$ are no closed subrelation of $I$.
		The lattice on the right corresponds to the context $\K_S$.
		The highlighted concepts $(\{4,6\},\{a,b,c\})$ and $(\{4,5,6\},\{a,c\})$ do not exist in the original lattice.
	}
	\label{fig:not_dismantlable}
\end{figure}

\Cref{thm:main} is a generalization of the following proposition concerning the dismantling of doubly irreducible lattice elements.

\begin{proposition}[Prop. 53 \cite{fca-book}]\label{prop:fca-doubly-irred}
  If $a= (g'',g') = (m',m'')$ is a doubly irreducible concept of a clarified context $\GMI$, then
  $$
  \BB(\GMI)\setminus \{a\} = \BB(G,M,I\setminus\{(g,m)\}).
  $$
\end{proposition}

\Cref{prop:connection-dismantling,prop:connection-quasi-dismantling} clarify how \cref{prop:fca-doubly-irred} and \cref{thm:main} are connected.

\begin{proposition}\label{prop:connection-dismantling}
	Let $\K=\GMI$ be a formal context and $\UBB(\K)$ its corresponding concept lattice.
	Let $\US\leq\UBB(\K)$ be an interval with $|\US|=1$.
	$\US$ is dismantling for $\UBB(\K)$ if and only if $\US$ is doubly irreducible.
\end{proposition}

\begin{proposition}\label{prop:connection-quasi-dismantling}
	Let $\K=\GMI$ be a formal context and $\UBB(\K)$ its corresponding concept lattice.
	Let $\US\leq\UBB(\K)$ be an interval with $|\US|=1$.
	$\US$ is quasi-dismantling if and only if 
	\begin{itemize}
		\item $\US$ is doubly irreducible or
		\item $\US=\top$ and $\US$ is supremum-irreducible or
		\item $\US=\bot$ and $\US$ is infimum-irreducible or
		\item $\US=\top=\bot$.
	\end{itemize} 
\end{proposition}

If we consider multiple intervals at once, one direction of \cref{lem:closed_subrelation} still holds:

\begin{lemma}
	Let $\K=\GMI$ be a formal context, $\US_1\ldots \US_k$ be intervals in $\UBB(\K)$.
	If $\US_1\ldots \US_k$ are quasi-dismantling, then $(G_{S_1\cup\ldots\cup S_k},M_{S_1\cup\ldots\cup S_k},I_{S_1\cup\ldots\cup S_k})$ is a closed subcontext of $\K$.
\end{lemma}
\begin{proof}
	We show the contraposition.
	Assume $(G_{S_1\cup\ldots\cup S_k},M_{S_1\cup\ldots\cup S_k},I_{S_1\cup\ldots\cup S_k})$ is no closed subcontext.
	By definition, we have $G_{S_1\cup\ldots\cup S_k}\subseteq G$, $M_{S_1\cup\ldots\cup S_k}\subseteq M$ and $I_{S_1\cup\ldots\cup S_k}\subseteq I$.
	Then there exists $c\in\BB(\K_{S_1\cup\ldots\cup S_k})$ with $c\not\in\BB(\K)$.
	Hence, there exist $x,y\in\BB(\K)$ such that $c=x\vee y$ or $c=x\wedge y$ in $\BB(\K_{S_1\cup\ldots\cup S_k})$.
	In case $c=x\vee y$ in $\BB(\K_{S_1\cup\ldots\cup S_k})$ there exists $z\in\BB(\K)$ such that $z=x\vee y$.
	Since $z\in S_1\cup\ldots\cup S_k$ we have $z\in S_i$ for some $i$.
	The rest follows analogous to the proof of \cref{lem:closed_subrelation}.
	\qed
\end{proof}

  However, not all closed subcontexts (and therefore sublattices of the corresponding concept lattice) can be obtained via a quasi-dismantling interval or a set of quasi-dismantling intervals for the corresponding lattice, see e.g. \cref{fig:closed-subrelation-no-dismantlable-interval}.

  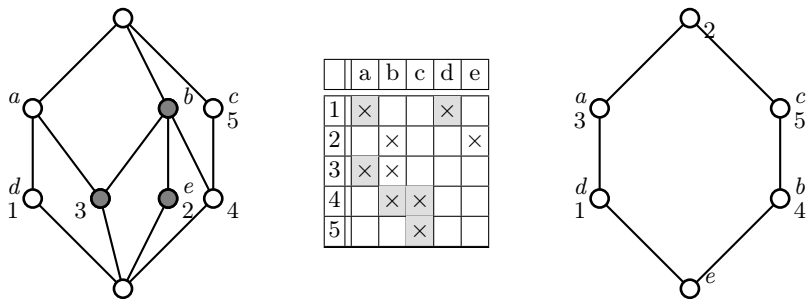
\begin{figure}[t]
    \centering
	\begin{minipage}{0.3\textwidth}
		\centering
      	{\unitlength 0.6mm
		\begin{picture}(40,60)%
		\put(0,0){%
			\begin{diagram}{40}{60}
			\Node{1}{20}{0}
			\Node{2}{0}{20}
			\Node{3}{15}{20}
      \Node{4}{30}{20}
      \Node{5}{40}{20}
      \Node{6}{0}{40}     
      \Node{7}{30}{40}     
      \Node{8}{40}{40}     
      \Node{9}{20}{60}     
      \Edge{1}{2}
      \Edge{1}{3}
      \Edge{1}{4}
      \Edge{1}{5}
      \Edge{2}{6}
      \Edge{3}{6}
      \Edge{3}{7}
      \Edge{4}{7}
      \Edge{5}{7}
      \Edge{5}{8}
      \Edge{6}{9}
      \Edge{7}{9}
      \Edge{8}{9}
      \leftObjbox{2}{3}{1}{1}
      \rightObjbox{4}{3}{1}{2}
			\leftObjbox{3}{3}{1}{3}
      \rightObjbox{5}{3}{1}{4}
      \rightObjbox{8}{3}{1}{5}
			\leftAttbox{6}{3}{1}{a}
			\rightAttbox{7}{3}{1}{b}
      \rightAttbox{8}{3}{1}{c}
      \leftAttbox{2}{3}{1}{d}
    	\rightAttbox{4}{3}{1}{e}	
    \end{diagram}}
			\put(15,20){\ColorNode{gray}}
			\put(30,20){\ColorNode{gray}}
			\put(30,40){\ColorNode{gray}}
		\end{picture}}	
	\end{minipage}
	\begin{minipage}{0.3\textwidth}
		\centering
		\begin{cxt}%
			\att{a}%
			\att{b}%
			\att{c}%
      \att{d}
      \att{e}
			\obj{c..c.}{1} %
			\obj{.x..x}{2} %
			\obj{cx...}{3} %
      \obj{.cc..}{4} %
      \obj{..c..}{5} %
		\end{cxt}
	\end{minipage}
	\begin{minipage}{0.3\textwidth} 
	\centering 
  	{\unitlength 0.6mm
		\begin{picture}(40,60)%
		\put(0,0){%
			\begin{diagram}{40}{60}
			\Node{1}{20}{0}
			\Node{2}{0}{20}
      \Node{3}{40}{20}
      \Node{4}{0}{40}        
      \Node{5}{40}{40}     
      \Node{6}{20}{60}     
      \Edge{1}{2}
      \Edge{1}{3}
      \Edge{2}{4}
      \Edge{3}{5}
      \Edge{4}{6}
      \Edge{5}{6}
      \leftObjbox{2}{3}{1}{1}
      \leftObjbox{4}{3}{1}{3}
      \rightObjbox{3}{3}{1}{4}
      \rightObjbox{5}{3}{1}{5}
      \rightObjbox{6}{3}{1}{2}
			\leftAttbox{2}{3}{1}{d}
      \leftAttbox{4}{3}{1}{a}
      \rightAttbox{3}{3}{1}{b}
      \rightAttbox{5}{3}{1}{c}
      \rightAttbox{1}{3}{1}{e}
      \end{diagram}}
		\end{picture}}	
\end{minipage} 
\caption{
  A closed subrelation (middle) that can not be obtained as via dismantling intervals.
  The interval $[(\{3\},\{a,b\}),(\{2,3,4\},\{b\})]$ is not dismantling, and neither $(\{3\},\{a,b\})$ nor $(\{2,3,4\},\{b\})$ are doubly irreducible.
}
\label{fig:closed-subrelation-no-dismantlable-interval}
\end{figure}

Further note, that not every lattice contains a quasi-dismantling interval besides the trivial one (the complete lattice) or any dismantling interval at all, see for example \cref{fig:lattice-without-DI}.

\begin{figure}[t]
  \centering
  	\begin{minipage}{0.3\textwidth} 
      \centering
    \scalebox{0.9}{
      {
        \unitlength 0.6mm
		\begin{picture}(75,60)%
		\put(0,0){%
			\begin{diagram}{75}{60}
			\Node{ 1}{37.5}{0}
			\Node{ 2}{0}{20}
      \Node{ 3}{15}{20}
      \Node{ 4}{30}{20}        
      \Node{ 5}{45}{20}     
      \Node{ 6}{60}{20}     
      \Node{ 7}{75}{20}
      \Node{ 8}{0}{40}
      \Node{ 9}{25}{40}
      \Node{10}{50}{40}
      \Node{11}{75}{40}
      \Node{12}{37.5}{60}
      \Edge{1}{2}
      \Edge{1}{3}
      \Edge{1}{4}
      \Edge{1}{5}
      \Edge{1}{6}
      \Edge{1}{7}
      \Edge{2}{8}
      \Edge{2}{9}
      \Edge{3}{8}
      \Edge{3}{10}
      \Edge{4}{8}
      \Edge{4}{11}
      \Edge{5}{9}
      \Edge{5}{10}
      \Edge{6}{9}
      \Edge{6}{11}
      \Edge{7}{10}
      \Edge{7}{11}
      \Edge{8}{12}
      \Edge{9}{12}
      \Edge{10}{12}
      \Edge{11}{12}
      \end{diagram}}
  \end{picture}
}
}
\end{minipage}
\caption{Smallest (non-trivial) lattice that has no dismantling interval.}
\label{fig:lattice-without-DI}
\end{figure}
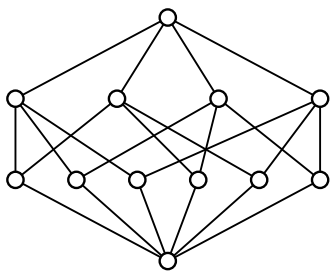

However, there is always a unique smallest lattice that can be obtained by iteratively removing all dismantling intervals, as shown in \cref{thm:unique-dismantling-core}.
To this end, we make use of the following lemma concerning the dismantlability of intervals upon removing one of them.

\begin{lemma}
  \label{lem:dismantlable-intervals-stay-dismantlable}
  Let $\UL$ be a lattice and $\US_1,\US_2$ dismantling intervals for $\UL$ such that $\US_2\not\subseteq\US_1$.
  Then, $\US_2\setminus\US_1$ is a dismantling interval for $\UL\setminus\US_1$.
\end{lemma}
\begin{proof}
  We first show that $\US_2\setminus\US_1$ is an interval:
  Assume $\US_2\setminus\US_1$ is no interval in $\UL\setminus\US_1$.
  Then, without loss of generality, there exist $x,y\in\US_2\setminus\US_1$ such that $z\coloneqq x \vee y \not\in \US_2\setminus\US_1$ in $\UL\setminus\US_1$.
  Either $z = x\vee y$ in $\UL$, thus $z\in\US_2$ and therefore $z\in\US_2\setminus\US_1$\contradiction; or $z\not= x\vee y$ in $\UL$, hence $w = x\vee y$ in $\UL$ with $w\in \US_1$ and since $\US_1$ is dismantling,  $x\in\US_1$ or $y\in \US_1$\contradiction.
  Thus, $\US_2\setminus\US_1$ is an interval in $\UL\setminus\US_1$.

  It remains to show that $\US_2\setminus\US_1$ is dismantling for $\UL\setminus\US_1$.
  Let $[u,v]=\US_2\setminus\US_1$ and assume $\US_2\setminus\US_1$ is not dismantling.
  Then, $u$ is not $\supprime$ in $(v]$ or $v$ is not $\infprime$ in $[u)$.
  Without loss of generality, assume $u$ is not $\supprime$ in $(v]$.
  Then, there exist $x,y\in\UL\setminus\US_1$ such that $x,y\not\in[u,v]$ and $x\vee y \in [u,v]$.
  Hence, $x\vee y\in\US_2$ in $\UL$.
  Because $\US_2$ is dismantling for $\UL$ it follows that $x\in\US_2$ or $y\in\US_2$ and therefore $u$ is $\supprime$\contradiction. Thus, $\US_2\setminus\US_1$ is dismantling in $\UL\setminus\US_1$.\qed
\end{proof}

Let $DI(\UL)$ be the family of all subsets of $\UL$ that can be obtained by iterated dismantling from $\UL$, \ie, by iteratively removing dismantling intervals starting from $\UL$.
A smallest element of $DI(\UL)$ is called a \emph{DI-core} of $\UL$.

\begin{theorem}\label{thm:unique-dismantling-core}
	Let $\UL$ be a lattice. There exists a unique \emph{DI-core}.
\end{theorem}
\begin{proof}
	Let $\underline{U},\underline{V}\in DI(\UL)$ be two minimal elements in $DI(\UL)$.
	Then, there is a minimal upper bound $\underline{T}\in DI(\UL)$ of $\underline{U}$ and $\underline{V}$, \ie, both $\underline{U}$ and $\underline{V}$ are obtained by removing dismantling intervals from $\underline{T}$.
	Hence, there are two sequences of intervals $\US_1,...,\US_k,\underline{R}_1,...,\underline{R_l}$ such that $\underline{U}=\underline{T}\setminus(\US_1\cup...\cup\US_k)=(\underline{T}\setminus\US_1)\ldots\setminus\US_k$ and $\underline{V} = \underline{T}\setminus (\underline{R}_1\cup...\cup\underline{R_l})=(\underline{T}\setminus\underline{R_1})\ldots\setminus\underline{R_l}$.
  By iterative application of \cref{lem:dismantlable-intervals-stay-dismantlable} we have that $\underline{T}\setminus (\US_1\cup...\cup\US_k\cup\underline{R}_1\cup...\cup\underline{R_l} )\in DI(\UL)$.
 \qed
\end{proof}

\section{Dismantling in the Formal Context}
\label{sec:dismantling-contexts-and-arrow-relations}
In this section we show that dismantling intervals can be identified directly in the formal context. 
Based on this, we propose an algorithm to find all dismantling intervals for a given formal context.
Thus, we can omit the (expensive) computation of the concept lattice.
To this end, we make use of the arrow relations from FCA:

\begin{definition}[Def. 25 \cite{fca-book}]
	Let $\K=\GMI$ be a formal context. 
	For an object $g\in G$ and an attribute $m\in M$ we write 
	\begin{align*}
	g\swarrow m~ & :\Leftrightarrow &&\begin{cases}
	(g,m)\not\in I \text{ and }\\
	\text{if } \exists \  h\in G \text{ with } g'\subseteq h' \text{ and } g'\not =h', \text{ then } (h,m)\in I,
	\end{cases}\\
	g\nearrow m~ & :\Leftrightarrow &&\begin{cases}
	(g,m)\not\in I \text{ and }\\
	\text{if } \exists\  n\in M \text{ with } m'\subseteq n' \text{ and } m'\not =n', \text{ then } (g,n)\in I,
	\end{cases}\\
	g\neswarrow m~ & :\Leftrightarrow &&g\swarrow m \text{ and } g\nearrow m.
	\end{align*}
\end{definition}

We now adapt the definitions of $(\cdot)^{\star},(\cdot)_{\star}$ to object and attribute concepts,
\begin{align*}
  (\gamma g)_{\star} &\coloneqq \bigvee \{c \in \BB(\K)\ |\ c \le \gamma g\},~\text{and}\\
(\mu m)^{\star} &\coloneqq \bigwedge \{c \in \BB(\K)\ |\ \mu m \le c\},
\end{align*}
in order to characterize the arrow relations in a formal context as follows~\cite{fca-book}:
\begin{align*}
g\swarrow m &\Leftrightarrow \gamma g \wedge \mu m = (\gamma g)_{\star}\not= \gamma g\\
g\nearrow m &\Leftrightarrow \gamma g \vee \mu m = (\mu m)^{\star}\not= \mu m
\end{align*}
If $(\gamma g)_{\star}\not= \gamma g$ then the object $g$ is irreducible in the formal context and $\gamma g$ is supremum-irreducible in the corresponding concept lattice.
Analogously, if $(\mu m)^{\star}\not= \mu m$ then the attribute $m$ is irreducible in the formal context and $\mu m$ is infimum irreducible in the concept lattice.
Therefore, the arrow relations can be utilized to identify the irreducible concepts of a lattice~\cite[Lemma 13]{fca-book}.
This statement can be adapted to sublattices and their corresponding subcontexts, in particular to filters and ideals.

\begin{lemma}\label{lem:irred-iff-updownarrow-in-cxt}
	Let $\K=\GMI$ with $g\in G$ and $m\in M$.
	Then:
	\begin{enumerate}
		\item $\gamma g$ supremum irreducible in $(\mu m ]$ $\Leftrightarrow$
		$\exists n\in M: g\neswarrow n$ in 
		$\K|_{m',M}$(clarified)
		\item $\mu m$ supremum irreducible in $[\gamma g )$ $\Leftrightarrow$
		$\exists h\in G: h\neswarrow m$ in $\K|_{G,g'}$ 
		(clarified).
	\end{enumerate}
\end{lemma}
\begin{proof}
  This follows directly from \cite[Lemma 13]{fca-book} with $\UBB(\K|_{m',M})=(\mu m]$ and $\UBB(\K|_{G,g'})=[\gamma g)$.\qed
\end{proof}

Based on this equivalence, we propose a characterization for supremum-prime and infimum-prime concepts in the formal context.

\begin{proposition}\label{prop:dismantling-in-the-context}
  Let $\K=\GMI$ be a formal context, $g\in G$, $m\in M$, then
  \begin{enumerate}
  \item   $\gamma g$ is supremum prime in $(\mu m]$ if and only if
  \begin{enumerate}[label=\roman*)]
    \item $\exists n\in M: g\neswarrow n$ in $\K|_{m',M}$ (attribute-clarified) and
    \item $\nexists n\not= k \in M: g \nearrow k$ in $\K|_{m',M}$ (attribute-clarified).
  \end{enumerate}
    
\item $\mu m$ is infimum prime in $[\gamma g)$ if and only if
  \begin{enumerate}[label=\roman*)]
	  \item $\exists h\in G: h\neswarrow m$ in $\K|_{G,g'}$ (object-clarified) and
	  \item $\nexists h\not= o \in M: o \nearrow m$ in $\K|_{G,g'}$ (object-clarified).
    \end{enumerate}
  \end{enumerate}
\end{proposition}
\begin{proof}
	We show the first part of the statement:
	
  ``$\Leftarrow$'': 
  We show this by contraposition.
  Assume $\gamma g$ is not supremum prime in $(\mu m ]$ but is supremum irreducible (otherwise use \Cref{lem:irred-iff-updownarrow-in-cxt}).
  Hence, $\exists c_1,c_2\in (\mu m ], c_1\not= c_2, \gamma g \not \leq c_1, c_2, \gamma g \leq c_1 \vee c_2$.
  Let $c_i=\mu k_i$ with $k_i\in M$ such that there is no $l\in M$ with $\mu l > \mu k_i$ and $\gamma g \not \leq \mu l$, \ie, choose maximal attribute concepts not larger than $\gamma g$.

  We show that $g\nearrow k_i$ using the characterization $\gamma g \vee \mu k_i = (\mu k_i)^{\star} \not= \mu k_i$.
  The second part, $(\mu k_i)^{\star} \not= \mu k_i$, is fulfilled by choice of $\mu k_i$.
  If we assume that $\gamma g \vee \mu k_i \not= (\mu k_i)^{\star}$, then $(\mu k_i)^{\star} < \gamma g \vee \mu k_i$ and thus there exists some $l\in M$ with $l\in\intent({(\mu k_i)^{\star}}), l \not\in\intent({\gamma g}), l\in\intent({\mu k_i})$\contradiction.
  Hence, $g\nearrow k_i$ in $\K|_{m',M}$.

  ``$\Rightarrow$'':
	We show this by contraposition.
  Assume $k,n\in M$, $k\not= n$, $g\nearrow n$, $g\nearrow k$ in $\K|_{m',M}$ (clarified).
  From $g\nearrow n$ we have $\gamma g \vee \mu n = (\mu n)^{\star}\not= \mu n$ and thus $\mu n \not \geq \gamma g$.
  Analogously, $g\nearrow k$ implies $\gamma g \vee \mu k = (\mu k)^{\star}\not= \mu k$ and thus $\mu k \not \geq \gamma g$.
  Since $\gamma g \vee \mu k = (\mu k)^{\star}$ and $\gamma g \vee \mu n = (\mu n)^{\star}$, we have $\mu k \not\geq \mu n$ and $\mu n \not \geq \mu k$.
  Thus, we have $\mu k \vee \mu n \geq \gamma g$.
  
  The second part of the statement can be shown analogously.\qed
\end{proof}

Now, two questions arise. First, given we have a formal context $\K$ and an interval $[\gamma g, \mu m]$ between an object concept and an attribute concept, is this interval dismantling in $\UBB(\K)$? And second, given a formal context $\K$, which are the dismantling intervals in the corresponding concept lattice $\UBB(\K)$?

To answer the first question \Cref{prop:dismantling-in-the-context} tells us that it suffices to check the arrow relations of $g$ in $\K|_{m',M}$ and of $m$ in $\K|_{G,g'}$:
If $g$ only has a single $\neswarrow$ in $\K|_{m',M}$ and no additional $\nearrow$, then $\gamma g$ is supremum-prime in $(\mu m]$. Analogously, if $m$ only has a single $\neswarrow$ in $\K|_{G,g'}$ and no additional $\swarrow$, then $\mu m$ is infimum-prime in $[\gamma g)$.
If both conditions hold, then the interval $[\gamma g, \mu m]$ is dismantling in $\UBB(\K)$.
Note that, if $\gamma g \not\leq \mu m$ then $g\not\in \K|_{m',M}$ and  $m\not\in \K|_{G,g'}$.

\begin{algorithm}[t]
  \SetAlgoLined
  \LinesNumbered
  \SetAlgoNoEnd
  \DontPrintSemicolon
   
\KwIn{A formal context $\K=\GMI$}
\KwResult{The set of all dismantling intervals for $\K$.}
\BlankLine 
$U=\emptyset$, 
$O=\emptyset$\;
\For{$g\in G$}{
  compute $\K|_{G,g'}$ and clarify objects\;
  compute $\nearrow(\K|_{G,g'})=\{(h,m)\ |\ h\nearrow m \text{ in } \K|_{G,g'}\}$\;
  \For{$m\in g'$}{
    $H_m= (G\times \{m\})\cap \nearrow(\K|_{G,g'})$\;
    \If{$H_m=\{(h,m)\}$ and $h\swarrow m$ in $\K|_{G,g'}$}{
      $U= U\cup \{(g,m)\}$\;
    }
  }
} 
\For{$m\in M$}{
  compute $\K|_{m',M}$ and clarify attributes\;
  compute $\swarrow(\K|_{m',M})=\{(g,n)\ |\ g\swarrow n \text{ in } \K|_{m'.M}\}$\;
  \For{$g\in m'$}{
    $N_g= (\{g\}\times M)\cap \swarrow(\K|_{m',M})$\;
    \If{$N_g=\{(g,n)\}$ and $g\nearrow n$ in $\K|_{m',M}$}{
      $O = O\cup \{(g,m)\}$\;
    }
  }
}
\Return $\{[\gamma g,\mu m] \ | \ (g,m) \in O\cap U\}$\;

\caption{Compute All Dismantling Intervals}
\label{alg:compute-all-dismantling-intervals} 
\end{algorithm}

In order to compute all dismantling intervals for a given formal context the naive approach is to check all intervals $[\gamma g, \mu m]$ between object concepts and attribute concepts.
However, this (essentially iterative) approach results in the repeated computation of the same subcontexts.
To prevent this, we instead compute each subcontext only once and for each object concept $\gamma g$ we check which attribute concepts are infimum-prime in $[\gamma g)$, \ie, we check the arrow relations in $\K|_{G,g'}$, and vice versa.
More precisely, for each object $g$ we take the attributes $m$ where $\mu m$ is infimum-prime in $[\gamma g)$ and collect them in the set
$$U= \{ (g,m)\ |\ g\in G, \mu m \text{ infimum-prime in } [\gamma g) \}.$$
Similarly, for each attribute $m$ we take the objects $g$ where $\gamma g$ is supremum-prime in $(\mu m]$ and collect them in the set
$$O= \{ (g,m)\ |\ m\in M, \gamma g \text{ supremum-prime in } (\mu m] \}.$$
If a pair $(g,m)$ is in both $U$ and $O$, then the respective interval $[\gamma g, \mu m]$ is dismantling for the lattice $\UBB(\K)$.
Note that it suffices to consider the reduced formal context.
In \Cref{alg:compute-all-dismantling-intervals} we present an implementation in pseudo-code.

If we are interested in the dismantling intervals of a lattice $\UL$, we can simply compute them for its standard context, \ie, for the context $(J(\UL),M(\UL),\leq)$ where $J(\UL)$ are the supremum-irreducible, and $M(\UL)$ the infimum-irreducible elements of $\UL$.

\section{Conclusion}
\label{sec:conclusion}
In this paper, we introduce the notion of dismantling intervals for a lattice in order to transfer the notion of dismantling doubly irreducible elements to a set of elements.
In particular, we show the connection between closed subrelations on the context side and dismantling intervals on the lattice side, and more generally the connection between closed subcontexts and quasi-dismantling intervals.
While a lattice can always be shrunk to the trivial empty lattice by removing a quasi-dismantling interval, iteratively removing only dismantling intervals for a lattice results in a unique (not necessarily trivial) smallest sublattice.
The dismantling intervals can be found directly in the formal context $\K$ with help of the arrow relations.
We show how to decide in $\K$ if a given interval is dismantling for $\UBB(\K)$.
Additionally, given $\K$, we propose an algorithm to compute all dismantling intervals of $\UBB(\K)$ without first computing the concept lattice itself.

\subsection*{Acknowledgements}
We thank Bernhard Ganter for the helpful discussions and fruitful suggestions. 

\bibliographystyle{splncs04}
\bibliography{paper}

\end{document}